\documentclass{amsart}
\usepackage{amssymb}
\usepackage{amsfonts}

\newtheorem{theorem}{Theorem}
\theoremstyle{plain}

\newtheorem{axiom}{Axiom}

\newtheorem{definition}{Definition}

\numberwithin{equation}{section}
\input{tcilatex}

\begin{document}
\title[Relative power]{The relative power and its invariance}
\author{Paolo Maria Mariano}
\address{DICeA, University of Florence, via Santa Marta 3, I-50139 Firenze
(Italy)}
\email{paolo.mariano@unifi.it}

\begin{abstract}
The relative power of actions in a Cauchy body suffering mutations due to
defect evolution is introduced. It is shown that its invariance under the
action of the Euclidean group over the ambient space and the material space
allows one to obtain (\emph{i}) the balance of standard and configurational
actions and (\emph{ii}) the identification of configurational ingredients
from a unique source.
\end{abstract}

\maketitle

\section{ \ \ }

Actions driving the evolution of defects in materials are commonly called
configurational because they are associated with processes mutating the
material structure of a body in a way that can be represented through
alterations of the reference macroscopic configuration. It seems that the
term `configurational' should be attributed to Nabarro (see remarks in \cite%
{Ericksen}). In a pioneering paper \cite{Eshelby}, Eshelby observed that, in
simple elastic bodies undergoing large deformations, the equations obtained
by means of horizontal variations of the bulk elastic energy -- they are the
variations generated by altering the reference place by using appropriate
diffeomorphisms -- are associated with the equilibrium of defects with
non-vanishing volume. The irreversible evolution of these defects is also
described through the introduction of peculiar driving forces. The analysis
of the evolution of point, line and surface defects (vacancies,
dislocations, interfaces, cracks) and the justification of the relevant
balances of the actions governing the equilibrium and the evolution have
been discussed largely in the subsequent literature. Various points of view
generated a burning debate about the nature of the balances of
configurational actions, governing equilibrium and possible evolution of
defects in simple and complex bodies. On one side it has been claimed that
the local configurational balance is just the projection through the inverse
motion of the Cauchy balance in terms of Piola-Kirchhoff stress, in absence
of dissipative driving forces \cite{MauTr, Mau, Mau95}. On another side the
fundamental independent nature of the balance of configurational forces has
been supported: such a balance has been postulated a priori in an abstract
way, then its (at the beginning) unspecified ingredients (Hamilton-Eshelby
stress and configurational bulk forces) have been identified in terms of
standard actions by means of an invariance requirement and the second law of
thermodynamics, the use of which presumes the assignment of the free energy 
\cite{GS, G, Gu}.

Further contributions to the debate are manifold\footnote{%
Examples are \cite{J, Menzel, Sam, Se, Sil, SB, Tru, AV, Ep, KD, St, St1,
St2} and references therein. Of course, I do not claim completeness of this
list which collects possible choices.}.

Essential differences between the balance of forces involving the first
Piola-Kirchhoff stress and the balance of configurational actions have been
pointed out indirectly by the results in an earlier paper, namely \cite%
{GMS-ARMA} (see also further remarks in \cite{GMS}), in the case of elastic
bodies undergoing large deformations. In fact, results in calculus of
variations help once more in addressing the discussion. Consider only the
balances in the bulk just for the sake of simplicity. For smooth minimizers
it is obviously true that, in absence of evolution governed by a driving
force, the configurational balance equations can be obtained by pulling back
in the reference place by means of the inverse motion the relevant balances
in terms of standard Piola-Kirchhoff stress. Different is the case of
irregular minimizers. They are common because existence theorems place
minimizers of the elastic energy in Sobolev spaces. Sobolev maps do not
admit always tangential derivatives. For this reason one cannot compute the
balance of forces in terms of Piola-Kirchhoff stress from the first
variation of the energy functional. The so-called horizontal variations are
admitted: they are variations which alter the reference place and lead to
the balance of configurational forces (at least the one not accounting for
driving force). Similar variations are also admissible on the actual place
of the body: under appropriate bounds for the derivatives of the energy (or
better of its policonvex representative) one finds the weak form of the
balance of forces in terms of Cauchy stress and proves also that such a
stress is locally $L^{1}$. Thus, for irregular minimizers the technique
based on the inverse motion mentioned above cannot be applied. The technique
based on horizontal variations has been later applied to various cases (see,
e.g., \cite{MauTr}).

In the ensuing sections, by restricting the attention to simple bodies, I
present a procedure based on $\mathbb{R}^{3}\ltimes SO\left( 3\right) $
invariance of a certain power that I call the \emph{relative power}. It
allows one to obtain (\emph{i}) the balance of both standard and
configurational actions and (\emph{ii}) the identification of
configurational ingredients from a unique source. The idea is based on the
definition of two virtual velocity fields $v$ and $w$ acting one over the
ambient space and the other over the space in which the material
configuration of the body is placed. The latter field is then pushed forward
on the ambient space, along the motion and the power performed by the
standard actions in the difference between $v$ and the image of $w$ is
evaluated. Such a power is supplemented by a power of the energy flux
generated by the possible disarrangements and permutation of defects that
are determined by the action of $w$ in a material space, a space endowed
with hits own energy. The sum of all these contributions is exactly the
functional that I call the relative power. Its definition is not exotic and
is not different in essence from the one of standard power. It reduces to
the standard expression of the power when the reference place is fixed once
and for all as it happens in standard continuum mechanics.

Neither surface and line defects, nor material complexity are accounted for.
They are matter of future work. Here the attention is focused only on the
basic skeletal idea.

Peculiarities of the ensuing developments are summarized below.

\begin{enumerate}
\item \emph{Use of the inverse motion is not required.}

\item \emph{No integral configurational balance is postulated.}

\item \emph{The integral balance of configurational forces and a
configurational balance of torques are derived and correspond to Killing
fields of the metric in the material space.}

\item \emph{The existence of a free energy density is postulated but the
list of state variables entering its constitutive structure is not specified
to a certain extent.}

\item \emph{No use is made of the mechanical dissipation inequality to
identify the purely mechanical part of configurational forces. In fact, the
identification follows directly from the procedure.}

\item \emph{The procedure does not require a variational structure and holds
in dissipative setting.}

\item \emph{A balance arising by the requirement of invariance of the
relative power under changes in observers corresponds in purely conservative
case to an integral version of N\"{o}ther theorem.}
\end{enumerate}

Differences and analogies with the two different points of view analyzed in 
\cite{Mau95} and \ \cite{G} (and developed in subsequent papers) are further
discussed in the last section.

\section{ \ \ \ }

The description of the standard kinematics of simple continuous deformable
bodies is so well known that it barely needs to be retold. The setting is
the classical space-time. A fit region $\mathcal{B}$ (more simply, an open,
connected set with Lipschitz boundary) of the standard ambient space $%
\mathbb{R}^{3}$ receives a body in its reference place. Each ensuing
configuration is reached in an isomorphic copy of $\mathbb{R}^{3}$,
indicated by $\mathbb{\hat{R}}^{3}$, by means of a transplacement, an
orientation preserving diffeomorphism $x\longmapsto y:=y\left( x\right) \in 
\mathbb{\hat{R}}^{3}$. The set $\mathcal{B}_{a}:=y\left( \mathcal{B}\right) $
is then the actual configuration (placement) of the body. The spatial
derivative of $x\longmapsto y$ is indicated by $F:=Dy\left( x\right) \in
Hom\left( T_{x}\mathcal{B},T_{y\left( x\right) }\mathcal{B}_{a}\right) $.
The positivity of the determinant of $F$ at each $x$ from $\mathcal{B}$,
i.e. $\det F>0$, is implied by the assumption that the generic
transplacement be orientation preserving. The additional requirement%
\begin{equation*}
\int_{\mathcal{B}}\tilde{f}\left( x,y\left( x\right) \right) \det Dy\left(
x\right) \text{ }dx\leq \int_{\mathbb{\hat{R}}^{3}}\sup_{x\in \mathcal{B}}%
\tilde{f}\left( x,z\right) \text{ }dz
\end{equation*}%
for all $\tilde{f}\in C_{0}^{\infty }(\mathcal{B}\times \mathbb{\hat{R}}%
^{3}) $\ is a global one-to-one condition allowing frictionless self-contact
of the boundary while still preventing self-penetration (see \cite{GMS}).

In representing motions, time come into play and one has%
\begin{equation*}
\left( x,t\right) \longmapsto y:=y\left( x,t\right) \in \mathbb{\hat{R}}^{3},%
\text{ \ \ }x\in \mathcal{B},\text{ \ }t\in \left[ 0,\bar{t}\right] ,
\end{equation*}%
with a presumption of sufficient smoothness in time, so that the velocity
field is defined by%
\begin{equation*}
\left( x,t\right) \longmapsto \dot{y}=\frac{d}{dt}y\left( x,t\right) \in 
\mathbb{\hat{R}}^{3},
\end{equation*}%
in the reference configuration.

Every subset $\mathfrak{b}$\ from $\mathcal{B}$ with non-vanishing `volume'
measure and the same regularity of $\mathcal{B}$ itself is called a \emph{%
part}. The set $\mathfrak{P}\left( \mathcal{B}\right) $ of all parts of $%
\mathcal{B}$ is an algebra with respect to the operations of meet and join
(see \cite{CV}).

Virtual velocity fields are defined over the ambient space and the reference
places:%
\begin{equation*}
x\in \mathcal{B},\text{ \ }t\in \left[ 0,\bar{t}\right] ,\text{ \ \ }\left(
x,t\right) \longmapsto v:=v\left( x,t\right) \in \mathbb{\hat{R}}^{3},\text{
\ \ }\left( x,t\right) \longmapsto w:=w\left( x,t\right) \in \mathbb{R}^{3},
\end{equation*}%
They are assumed to be differentiable in space at every instant. The symbols 
$V_{v}$ and $V_{w}$ denote the functional spaces containing them. Elements
from $V_{v}$\ and $V_{w}$\ can be considered as virtual velocity fields over
the body.

In the previous picture, the generic material element is collapsed just in a
point which is its sole morphological descriptor. I use to call Cauchy
bodies those bodies for which the minimalist approach summarized above is
sufficient to represent the main essential peculiarities of their
morphology, the representation of actions is then conjugated in terms of
power. Different is the case of complex bodies for which descriptors of the
material substructure selected in a differentiable manifold are included in
the representation of the morphology of the generic material element%
\footnote{%
See \cite{C, CV, dFM, Mar}, \cite{Mar} and references therein.}.

\section{ \ \ }

An observer is a representation of all geometrical environments that are
necessary to describe the morphology of a body and its motion\footnote{%
Such a definition has non-trivial consequences above all in the mechanics of
complex bodies, rather than in the one of simple bodies (see references in
footnote 2).}. Here an observer is then a triple of atlas, one over the
ambient space $\mathbb{\hat{R}}^{3}$, one over the material space containing 
$\mathcal{B}$ and the last one over the time interval. Changes in observers
are then changes in these atlas, governed by the relevant groups of
diffeomorphisms. In particular I consider synchronous isometric changes in
observers. Synchronicity means that the representation of the time interval
is left invariant.

By indicating by $v^{\ast }$ the pull back in the first observer of the rate
measured by the second observer, the action of the semi-direct product $%
\mathbb{\hat{R}}^{3}\ltimes SO\left( 3\right) $ over the ambient space $%
\mathbb{\hat{R}}^{3}$ gives rise to the standard formula%
\begin{equation*}
v^{\ast }=\hat{c}\left( t\right) +\hat{q}\left( t\right) \times \left(
y-y_{0}\right) +v,
\end{equation*}%
where $y_{0}$ is an arbitrary point in $\mathbb{\hat{R}}^{3}$, $\hat{c}%
\left( t\right) \in \mathbb{\hat{R}}^{3}$ and $\hat{q}\left( t\right) \times
\in so\left( 3\right) $, with $so\left( 3\right) $\ the Lie algebra of $%
SO\left( 3\right) $. In standard approaches, it is then assumed that all
observers evaluate the same $\mathcal{B}$.

Here the assumption is removed and the independent action of the semi-direct
product $\mathbb{R}^{3}\ltimes SO_{\diamond }\left( 3\right) $, with $%
SO_{\diamond }\left( 3\right) $ a copy isomorphic to $SO\left( 3\right) $,
over $\mathbb{R}^{3}$ is considered. It leads to the formula\footnote{%
Such a point of view has been recently also discussed in \cite{Mur} (see
also reference therein) for different purposes. It is also used in \cite{Gu}
with strict reference to the derivation of configurational balances. In
fact, a requirement of invariance of an expression of a power with respect
to such changes in observers is called upon. The power selected involves a
number of configurational actions (stresses, cinternal and external bulk
forces and couples) in an abstract way, without discussing at that stage
their possible expression in terms of standard actions. This point is
further analized in the last section.}%
\begin{equation*}
w^{\ast }=c\left( t\right) +q\left( t\right) \wedge \left( x-x_{0}\right) +w,
\end{equation*}%
where $x_{0}$ is an arbitrary point in $\mathbb{R}^{3}$, $c\left( t\right)
\in \mathbb{R}^{3}$ and $q\left( t\right) \times \in so_{\diamond }\left(
3\right) $, with $so_{\diamond }\left( 3\right) $\ the Lie algebra of $%
SO_{\diamond }\left( 3\right) $. The transformation $w\longmapsto w^{\ast }$
can be also considered as an isometric shift superposed to a generic
relabeling in the `material space' with infinitesimal generator $w$.

\section{ \ \ \ }

Surface and bulk actions are associated with (generated by) relative changes
of places between neighboring material elements: at every $x\in \mathcal{B}$
they are represented respectively by the first Piola-Kirchhoff stress $P\in
Hom(T_{x}^{\ast }\mathcal{B},T_{y\left( x\right) }^{\ast }\mathcal{B}%
_{a})\simeq \mathbb{\hat{R}}^{3\ast }\otimes \mathbb{R}^{3}$ and the vector
of bulk forces $b\in \mathbb{\hat{R}}^{3\ast }$ which includes inertial
actions when they are present.

The standard power of external actions over a generic part $\mathfrak{b}$ is
given by the expression%
\begin{equation*}
\mathcal{P}_{\mathfrak{b}}^{ext}\left( \dot{y}\right) :=\int_{\mathfrak{b}%
}b\cdot \dot{y}\text{ }dx+\int_{\partial \mathfrak{b}}Pn\cdot \dot{y}\text{ }%
d\mathcal{H}^{2}.
\end{equation*}%
Note that this expression is usually written by imagining that the reference
place does not undergo mutations. The requirement of invariance of $\mathcal{%
P}_{\mathfrak{b}}^{ext}\left( \dot{y}\right) $, under changes in observers
leaving invariant $\mathcal{B}$\ and altering isometrically the ambient
space, furnishes integral and then pointwise balance equations \cite{N}.

Here the point of view is different: the body can mutate its material
structure. The world `mutation' needs mechanical definition. I do not
consider any specific mechanism of mutation. Rather, I account for the
indirect effects of classes of mutations: energy fluxes in the material,
bulk driving forces and configurational couples. All these ingredients are
pictured in $\mathcal{B}$. They can be considered as due to the
rearrangements of possible inhomogeneities, their possible evolution and/or
to more general alterations of the material structure that can be pictured
through mutations of the reference placement $\mathcal{B}$. An extended
notion of power is then required. I call it a \emph{relative power}: it is
the power of standard actions evaluated on the velocity relative to the
rates of mutations in the reference place, supplemented by the energy fluxes
and the power of driving forces and configurational couples. The definition
of the relative power is presented after necessary ensuing preliminaries.

A free energy density $e$ is defined over $\mathcal{B}$; it is function of
the state $\varsigma $, the place $x$ and the time $t$, namely%
\begin{equation*}
e:=e\left( x,t;\varsigma \right) .
\end{equation*}

The state $\varsigma $\ of a material element is not specified here. The
explicit (direct) dependence on $x$ underlines the assumption that the
material is not homogeneous.\ The explicit dependence on time may describes
only some aspects of possible mutations, for example aging. In fact, for
elastic bodies with time-dependent moduli, the Clausius-Duhem inequality in
its isothermal version implies $\partial _{t}e\leq 0$ which corresponds
exactly to aging in time. In what follows the derivative $\partial _{x}e$
can be considered as the explicit derivative with respect to $x$, holding
fixed the state.

For the sake of simplicity, I do not consider below the explicit dependence
on time, so that from now on the free energy depends on the place $x$ and
the state $\varsigma $.

Standard tractions and bulk forces arise during a generic motion. They are
power-conjugated with the rate of changes of (relative) places of material
elements. They contribute to the equilibrium of defects and their evolution.
In presence of evolving structural \emph{mutations in the bulk},
annihilation and creation of material bonds occur. Bulk actions are then
power-conjugated with mechanisms of annihilation and creation (or
restoration) of material bonds. A bulk force $f$\ is then associated with
evolving mutations: it is the so-called driving force. Effects of anisotropy
in the distribution of mutations and anisotropies induced by the
`permutations' of defects in the bulk are accounted for by body couples $\mu 
$. By definition $\mu =0$ when anisotropies are absent in the material
space.\ Driving forces are introduced and justified variously in the
literature (see \cite{AK}). Configurational bulk couples have been
introduced in \cite{Gu} with the same meaning adopted here. Both $f$ and $%
\mu $ are described by co-vectors over $\mathcal{B}$\ because they are
associated with mechanisms mutating $\mathcal{B}$ itself. No configurational
traction associated with a primitive configurational stress is presumed a
priori: it is found later as a derived `object'.

Inertia is neglected here for the sake of simplicity. It can be included by
considering the bulk forces decomposed additively into inertial and
non-inertial parts and `adding' to $e$ the kinetic energy.

\begin{definition}
For Cauchy bodies a linear functional $\mathcal{P}^{rel}:\mathfrak{P}\left( 
\mathcal{B}\right) \times V_{v}\times V_{w}\rightarrow \mathbb{R}$ is called
a \textbf{relative power} when it is additive over disjoint parts, is linear
over the space of rates and admits the explicit expression%
\begin{equation*}
\mathcal{P}_{\mathfrak{b}}^{rel}\left( v,w\right) :=\mathcal{P}_{\mathfrak{b}%
}^{rel-a}\left( v,w\right) +\mathcal{P}_{\mathfrak{b}}^{dis}\left( v,w\right)
\end{equation*}%
with%
\begin{equation*}
\mathcal{P}_{\mathfrak{b}}^{rel-a}\left( v,w\right) :=\int_{\mathfrak{b}%
}b\cdot \left( v-Fw\right) \text{ }dx+\int_{\partial \mathfrak{b}}Pn\cdot
\left( v-Fw\right) \text{ }d\mathcal{H}^{2},
\end{equation*}%
\begin{eqnarray*}
\mathcal{P}_{\mathfrak{b}}^{dis}\left( v,w\right) &:&=\int_{\partial 
\mathfrak{b}}\left( n\cdot w\right) e\text{ }d\mathcal{H}^{2}+\int_{%
\mathfrak{b}}\left( \partial _{x}e-f\right) \cdot \left( w-\func{curl}%
w\times \left( x-x_{0}\right) \right) \text{ }dx+ \\
&&+\int_{\mathfrak{b}}\mu \cdot \func{curl}w\text{ }dx
\end{eqnarray*}
\end{definition}

I call $\mathcal{P}_{\mathfrak{b}}^{rel-a}\left( v,w\right) $ the \emph{%
relative power of actions} and $\mathcal{P}_{\mathfrak{b}}^{dis}\left(
v,w\right) $ the\emph{\ power due to disarrangements}.

\begin{enumerate}
\item The power of actions is said to be relative because it is developed
along the difference between the actual velocity and the push forward of the
material velocity $w$ in $\mathcal{B}_{a}$.

\item More difficult is the interpretation of the terms in $\mathcal{P}_{%
\mathfrak{b}}^{dis}\left( v,w\right) $. Recall that the velocity field $%
\left( x,t\right) \longmapsto v$ moves just points in space where no
material elements are necessarily placed. The one-parameter group of
diffeomorphism associated with the field $\left( x,t\right) \longmapsto w$
alters the distribution of the material elements, even permuting them in a
virtual way (it has the same role of the relabeling in calculus of
variation). A flux of energy through the boundary $\partial \mathfrak{b}$
appears. Moreover, the distribution of energy in space can be in principle
inhomogeneous. As mentioned above, both $x\longmapsto f\left( x\right) $ and 
$x\longmapsto \mu \left( x\right) $ are co-vectors fields over $\mathcal{B}$%
, thus material interactions power conjugated with the rate of change of the
inhomogeneities. Contrary, all standard forces are co-vectors over $\mathcal{%
B}_{a}$. Remind that couples $\mu $ are associated with the anisotropy
induced by the material mutations, including the permutation of defects. The
presence of mutations allow also one to include in the scenario the driving
force $f$. All effects associated with anisotropies both in the evolving
mutations and in the distribution of the energy as a consequence of the
relabeling are all included in $\mu $. Such a remark justifies the term $%
\left( w-\func{curl}w\times \left( x-x_{0}\right) \right) $. The negative
sign before $f$ is there only for the sake of convenience.
\end{enumerate}

Take note that $v$ may coincide with the true velocity $\dot{y}$ at $x$ and $%
t$.

Further physical justifications of Definition 1 are presented later.

\begin{axiom}
$\mathcal{P}_{\mathfrak{b}}^{rel}\left( v,w\right) $ is invariant under
isometric changes in observers.
\end{axiom}

All observers `measure' the same value of the power which is a scalar. The
axiom is not different in intrinsic meaning from the axiom of invariance of
the standard power \cite{N}. Differences in the expression of the power are
dictated only by the situation under scrutiny. Consequences are summarized
in the ensuing theorem which is the main result of this paper.

\begin{theorem}
$(i)$ If the fields $x\longmapsto b:=b\left( x\right) $ and $x\longmapsto
P:=P\left( x\right) $ are integrable over $\mathcal{B}$, then for every part 
$\mathfrak{b}$ the following integral balances hold:%
\begin{equation*}
\int_{\mathfrak{b}}b\text{ }dx+\int_{\partial \mathfrak{b}}Pn\text{ }d%
\mathcal{H}^{2}=0,
\end{equation*}%
\begin{equation*}
\int_{\mathfrak{b}}\left( y-y_{0}\right) \times b\text{ }dx+\int_{\partial 
\mathfrak{b}}\left( y-y_{0}\right) \times Pn\text{ }d\mathcal{H}^{2}=0,
\end{equation*}%
\begin{equation*}
\int_{\partial \mathfrak{b}}\mathbb{P}n\text{ }d\mathcal{H}^{2}-\int_{%
\mathfrak{b}}F^{\ast }b\text{ }dx+\int_{\mathfrak{b}}\left( \partial
_{x}e-f\right) \text{ }dx=0,
\end{equation*}%
\begin{equation*}
\int_{\partial \mathfrak{b}}\left( x-x_{0}\right) \times \mathbb{P}n\text{ }d%
\mathcal{H}^{2}-\int_{\mathfrak{b}}\left( x-x_{0}\right) \times F^{\ast }b%
\text{ }dx+\int_{\mathfrak{b}}\mu \text{ }dx=0.
\end{equation*}%
where, with $I$ the second order unit tensor,%
\begin{equation*}
\mathbb{P}:=eI-F^{\ast }P.
\end{equation*}%
$(ii)$ If the fields $x\longmapsto P$ and $x\longmapsto \mathbb{P}$ are of
class $C^{1}\left( \mathcal{B}\right) \cap C^{0}\left( \mathcal{\bar{B}}%
\right) $ then%
\begin{equation*}
DivP+b=0,
\end{equation*}%
\begin{equation*}
SkwPF^{\ast }=0,
\end{equation*}%
\begin{equation*}
Div\mathbb{P}-F^{\ast }b+\partial _{x}e=f.
\end{equation*}%
\begin{equation*}
2Skw\mathbb{P}=\mu \times
\end{equation*}%
$(iii)$ If the material is homogeneous and no driving force is present, when
the material distribution of defects and energy is isotropic in $\mathcal{B}$%
, then $\mathbb{P}$ is symmetric and, in absence of body forces, the integral%
\begin{equation*}
\int_{\partial \mathfrak{b}}\mathbb{P}n\text{ }d\mathcal{H}^{2}
\end{equation*}%
is `surface' independent. $(iv)$ An extended version of the virtual power
principle holds:%
\begin{equation*}
\mathcal{P}_{\mathfrak{b}}^{rel}\left( v,w\right) =\mathcal{P}_{\mathfrak{b}%
}^{rel-inn}\left( v,w\right) ,
\end{equation*}%
where%
\begin{equation*}
\mathcal{P}_{\mathfrak{b}}^{rel-inn}\left( v,w\right) :=\int_{\mathfrak{b}%
}\left( P\cdot \nabla v+\mathbb{P}\cdot \nabla w-\left( x-x_{0}\right)
\otimes \left( \partial _{x}e-f\right) \cdot Skw\nabla w+\mu \cdot \func{curl%
}w\right) \text{ }dx;
\end{equation*}%
it reduces to 
\begin{equation*}
\int_{\mathfrak{b}}b\cdot \left( v-Fw\right) \text{ }dx+\int_{\partial 
\mathfrak{b}}Pn\cdot \left( v-Fw\right) \text{ }d\mathcal{H}%
^{2}+\int_{\partial \mathfrak{b}}\left( n\cdot w\right) e\text{ }d\mathcal{H}%
^{2}+\int_{\mathfrak{b}}\left( \partial _{x}e-f\right) \cdot w\text{ }dx=
\end{equation*}%
\begin{equation*}
=\int_{\mathfrak{b}}P\cdot \left( \nabla v-F\nabla w\right) \text{ }dx.
\end{equation*}
\end{theorem}

\begin{proof}
The axiom of invariance and some elementary algebra impose that%
\begin{equation*}
\hat{c}\cdot (\int_{\mathfrak{b}}b\text{ }dx+\int_{\partial \mathfrak{b}}Pn%
\text{ }d\mathcal{H}^{2})+
\end{equation*}%
\begin{equation*}
+\hat{q}\cdot (\int_{\mathfrak{b}}\left( y-y_{0}\right) \times b\text{ }%
dx+\int_{\partial \mathfrak{b}}\left( y-y_{0}\right) \times Pn\text{ }d%
\mathcal{H}^{2})+
\end{equation*}%
\begin{equation*}
+c\cdot (\int_{\partial \mathfrak{b}}\left( eI-F^{\ast }P\right) n\text{ }d%
\mathcal{H}^{2}-\int_{\mathfrak{b}}F^{\ast }b\text{ }dx+\int_{\mathfrak{b}%
}\left( \partial _{x}e-f\right) \text{ }dx)+
\end{equation*}%
\begin{equation*}
+q\cdot (\int_{\partial \mathfrak{b}}\left( x-x_{0}\right) \times \mathbb{P}n%
\text{ }d\mathcal{H}^{2}-\int_{\mathfrak{b}}\left( x-x_{0}\right) \times
F^{\ast }b\text{ }dx+\int_{\mathfrak{b}}\mu \text{ }dx)=0.
\end{equation*}%
The arbitrariness of $c$, $q$, $\hat{c}$ and $\hat{q}$ implies the integral
balances in Theorem 1, once one defines $\mathbb{P}:=eI-F^{\ast }P$. The
poin-twise balances follow by the application of Gauss theorem. They imply
the equality between $\mathcal{P}_{\mathfrak{b}}^{rel}\left( v,w\right) $
and $\mathcal{P}_{\mathfrak{b}}^{rel-inn}\left( v,w\right) $. The last
statement of the theorem then follows straight away. Remind that $\mu =0$
when anisotropies are absent in $\mathcal{B}$.
\end{proof}

In the earlier theorem, $\mathbb{P}$ is called the Hamilton-Eshelby stress.
As recalled in the preamble, the word `configurational' is attributed to
balances involving it. Besides its immediateness, the earlier theorem has
some stringent theoretical consequences, as anticipated in the preamble.

\begin{enumerate}
\item To obtain the balance of configurational forces it is not necessary to
make use of the procedure exploiting the inverse motion.

\item The Hamilton-Eshelby stress $\mathbb{P}:=eI-F^{\ast }P$\ and the bulk
actions $-F^{\ast }b$\ and $\partial _{x}e$\ are not introduced a priori as
unknown objects and then identified with a procedure discussed further in
the last section.

\item A version of the principle of virtual power different from usual
arises, it includes the standard one when the reference place is considered
invariant, invariance intended in the sense of absence of evolving defects.
\end{enumerate}

The result can be extended to the case of complex bodies and to the case in
which structured discontinuity surfaces and line defects are present. In the
process, appropriate additional measures of interactions need to be
introduced as objects power conjugated with the rate of change of the
morphological descriptors of the substructure in the material elements (in
the case of complex bodies) and/or deformations and evolution of surface and
line defects. Even in that cases the procedure avoids the specification of
the constitutive structure of the local state and the use to the mechanical
dissipation inequality to identify the expression of the configurational
forces in terms of standard measures of interaction.

\section{ \ \ }

To explain further on the nature of the relative power, one may notice that
in purely conservative case the equation%
\begin{equation*}
\mathcal{P}_{\mathfrak{b}}^{rel}\left( v,w\right) =\mathcal{P}_{\mathfrak{b}%
}^{rel-inn}\left( v,w\right)
\end{equation*}%
reduces to an integral version of the pointwise balance appearing in N\"{o}%
ther theorem.

To prove such a statement consider a non-linear elastic inhomogeneous body
with total energy given by%
\begin{equation*}
e\left( x,F\right) +u\left( y\right) ,
\end{equation*}%
with $e$ the elastic potential -- a function which is polyconvex in the
gradient of deformation -- and $u$ the potential of body forces. Both $e$
and $u$ are assumed to be differentiable with respect to their arguments.
The essential ingredients preparing N\"{o}ther theorem need also to be
recalled briefly.

Consider smooth curves $s\longmapsto \mathbf{f}_{s}$ on the group of
diffeomorphisms $Diff(\mathbb{\hat{R}}^{3},\mathbb{\hat{R}}^{3})$ such that $%
\mathbf{f}_{0}=identity$ and at every point in $\mathbb{\hat{R}}^{3}$ one
gets $v=\frac{d}{ds}\mathbf{f}_{s}\left\vert _{s=0}\right. $, where the
field $y\longmapsto v\left( y\right) $ coincides with the virtual velocity
field introduced above over the ambient space.

The usual relabeling of the reference place is accounted for in $\mathbb{R}%
^{3}$. From a physical point of view it reduces just to the permutation of
inhomogeneities over $\mathcal{B}$. The relabeling is induced by the action
of the special group of diffeomorphisms $SDiff\left( \mathbb{R}^{3},\mathbb{R%
}^{3}\right) $, a group on which one selects smooth curves $s_{1}\longmapsto 
\mathbf{f}_{s_{1}}^{1}$ such that $\mathbf{f}_{0}^{1}=identity$ and at every 
$x$ one gets $w=\frac{d}{ds_{1}}\mathbf{f}_{s_{1}}\left\vert
_{s_{1}=0}\right. $, where the field $x\longmapsto w\left( x\right) $ is a
special case \ of the virtual velocity field $w$ introduced earlier, special
in the sense that it is isochoric.

Balance equations are obtained by requiring the minimum of the overall energy%
\begin{equation*}
\mathcal{E}\left( y\right) :=\int_{\mathcal{B}}\left( e\left( x,F\right)
+u\left( y\right) \right) \text{ }dx
\end{equation*}%
over an appropriate Sobolev space (commonly such a space is some $W^{1,p}$
or, more precisely, the space of weak diffeomorphisms discussed in \cite%
{GMS-ARMA, GMS}). In the case of $C^{1}$ minimizers, poin-twise
Euler-Lagrange equations can be derived in standard way. They read%
\begin{equation*}
b+DivP=0,
\end{equation*}%
where now $b:=-\partial _{y}u\in \mathbb{\hat{R}}^{3\ast }$ and $P:=\partial
_{F}e\in Hom\left( T_{x}^{\ast }\mathcal{B},T_{y\left( x\right) }^{\ast }%
\mathcal{B}_{a}\right) $. In the case of irregular minimizers one cannot
obtain the previous equation because Sobolev maps do not admit in general
tangential derivatives. The balance of configurational forces (discussed
previously) arises as a consequence of the evaluation of the horizontal
variations -- the ones generated by altering $\mathcal{B}$ through $%
s_{1}\longmapsto \mathbf{f}_{s_{1}}^{1}$. Moreover, one may obtain in
distributional sense the balance of forces in terms of Cauchy stress $\sigma
:=\left( \det F\right) ^{-1}\partial _{F}eF^{-\ast }\in \mathbb{\hat{R}}%
^{3\ast }\otimes \mathbb{\hat{R}}^{3\ast }$ and may prove also that $%
y\longmapsto \sigma $ belongs to $L_{loc}^{1}$ (see \cite{GMS}).

By focusing the attention for the sake of simplicity on the Euler-Lagrange
equations above, if one defines the vector density%
\begin{equation*}
\mathfrak{F}:=\left( e+u\right) w+\partial _{F}e^{\ast }\left( v-Fw\right) ,
\end{equation*}%
if the total energy is equivariant with respect to the action of $Diff(%
\mathbb{\hat{R}}^{3},\mathbb{\hat{R}}^{3})$ and $SDiff\left( \mathbb{R}^{3},%
\mathbb{R}^{3}\right) $, then (N\"{o}ther theorem, see e.g. \cite{KS, MH})%
\begin{equation*}
Div\mathfrak{F}=0.
\end{equation*}%
Equivariance means that%
\begin{equation*}
e\left( x,F\right) +u\left( y\right) =e(\mathbf{f}_{s_{1}}^{1}\left(
x\right) ,\left( grad\mathbf{f}_{s}\left( y\right) \right) F(\nabla \mathbf{f%
}_{s_{1}}^{1}\left( x\right) )^{-1})+u\left( \mathbf{f}_{s}\left( y\right)
\right) ,
\end{equation*}%
where $grad$ is the gradient with respect to $y$. The previous relation is
verified when (N\"{o}ther conditions)%
\begin{equation*}
\frac{d}{ds}(e(\mathbf{f}_{s_{1}}^{1}\left( x\right) ,\left( grad\mathbf{f}%
_{s}\left( y\right) \right) F(\nabla \mathbf{f}_{s_{1}}^{1}\left( x\right)
)^{-1})+u\left( \mathbf{f}_{s}\left( y\right) \right) )\left\vert
_{s=0}\right. =0,
\end{equation*}%
\begin{equation*}
\frac{d}{ds_{1}}(e(\mathbf{f}_{s_{1}}^{1}\left( x\right) ,\left( grad\mathbf{%
f}_{s}\left( y\right) \right) F(\nabla \mathbf{f}_{s_{1}}^{1}\left( x\right)
)^{-1})+u\left( \mathbf{f}_{s}\left( y\right) \right) )\left\vert
_{s_{1}=0}\right. =0.
\end{equation*}%
Such conditions read explicitly%
\begin{equation*}
\partial _{y}u\cdot v+\partial _{F}e\cdot \nabla v=0,
\end{equation*}%
\begin{equation*}
\partial _{x}e\cdot w-\partial _{F}e\cdot F\nabla w=0.
\end{equation*}%
By taking into account that (\emph{i}) in this case the field $x\longmapsto
w\left( x\right) $ is isochoric, namely $divw=0$, and (\emph{ii}) absence of
dissipative effects implies $f=0$, by considering also (\emph{iii}) the
explicit form of N\"{o}ther conditions and (\emph{iv}) the pointwise balance
of forces in terms of Piola-Kirchhoff stress, one realizes (after some
algebra) that the last relation in Theorem 1 reduces to the integral version
of N\"{o}ther theorem%
\begin{equation*}
\int_{\mathfrak{b}}\mathfrak{F}\cdot n=0
\end{equation*}%
on some arbitrary part $\mathfrak{b}$ of $\mathcal{B}$.

Conversely, one can say that the presence of a principle of relative power
is hidden in N\"{o}ther theorem. In fact, by starting from N\"{o}ther
theorem, I have already introduced in earlier papers, namely \cite{dFM, M,
Mar}, a version of the relative power including constitutive issues but also
surfaces and lines of discontinuity, without being conscious at that time of
its generality in non-conservative setting. This note, in fact, has
primarily the aim to stress this point.

\section{ \ \ }

The approach to configurational forces presented in \cite{Mau, Mau95, MauTr}
relies on constitutive assumptions. They are called upon only partially in
this paper: the sole assumption $e:=e\left( x,t;\varsigma \right) $ is
invoked without specifying the nature of the state $\varsigma $.

The comparison with the approach proposed in \cite{GS, G} (see also \cite{Gu}%
) requires a rather extended preliminary description. That approach is based
on two steps: (1) The balance of configurational forces is postulated first.
Such a postulate can be expressed trough the statement of an independent
integral balance (like in \cite{G}) or by requiring the invariance of a
certain power (a power which is different from the one used here) with
respect to the transformation $w\longmapsto w^{\ast }$ (like in \cite{Gu};
see also \cite{Se}). In \cite{G} and \cite{Gu}, independently of its origin,
the balance of configurational forces involves a configurational stress, say 
$\mathbb{P}$, and external and internal configurational bulk forces, say $%
\mathsf{\tilde{g}}$ and $\mathsf{\tilde{e}}$ respectively -- they are
different from the driving force $f$ and the configurational couple $\mu $
which play also a role in the treatment. The bulk configurational forces $%
\mathsf{\tilde{g}}$ and $\mathsf{\tilde{e}}$ are assumed to perform work
only after a Galilean change in observer (i.e. $w\longmapsto w+c$) at a
first glance (see \cite{Gu}, page 36). The point of view is then changed (%
\cite{Gu}, page 39) by saying that only $\mathsf{\tilde{g}}$ does not
perform power under time-dependent changes in reference is involved. Whether 
$\mathbb{P}$, $\mathsf{\tilde{g}}$ and $\mathsf{\tilde{e}}$ can be expressed
in terms of standard actions and energy is not known at this stage. The
identification of $\mathbb{P}$, $\mathsf{\tilde{g}}$ and $\mathsf{\tilde{e}}$
is matter of the second step. (2) It is essentially based on the
exploitation of the second law of thermodynamics written in terms of a
mechanical dissipation inequality in which only the power of the
configurational traction $\mathbb{P}n$ is added to the one of standard
actions. The mechanical dissipation inequality is written with respect to
control volumes with boundaries evolving in time. Invariance with respect to
the reparametrization of such boundaries leads to the identification $%
\mathbb{P}:=eI-F^{\ast }P$. Note that in the mechanical dissipation
inequality the energy is introduced. It is assumed also that $e$ is
differentiable with respect to time. The identification of $\mathsf{\tilde{e}%
}$ with $-F^{\ast }b$ follows directly from the insertion of the expression $%
eI-F^{\ast }P$ in the balance of configurational forces. The mechanical
dissipation inequality has to be exploited to recognize that $\mathsf{\tilde{%
g}}$\ coincides with $-\nabla e+P:\nabla F$. The additional specification of
the constitutive structure of the energy shows that $\mathsf{\tilde{g}}$
reduces finally to $-\partial _{x}e$ (see \cite{G}, page 78).

Comparison of the results in \cite{G}\ with the point of view in the
previous sections has to be done at the end of the identification procedure
(so that after step 2) because not only Theorem 1 collects the balance of
configurational forces but also it \emph{includes} the results of the
identification recalled above, at least in the setting discussed here.

To obtain Theorem 1 -- I stress once more -- no use is made of the
mechanical dissipation inequality. No use is made of an additional
requirement of invariance of the power with respect to the reparametrization
of $\partial \mathfrak{b}$. Although the energy is also introduced here, no
assumption of differentiability in time is necessary.

The state here is not specified: for example it can include $F$, the history
of deformation, a number of internal variables conjugated with affinities,
their histories and gradients. The sole restriction is that the state be
compatible with the relative power of actions. In fact, higher-order Cauchy
bodies (like, e.g., second-grade elastic bodies) or complex bodies require
expressions of $\mathcal{P}^{rel-a}$ involving hyperstresses or
microstresses and self-actions respectively.

I do not claim that the treatment proposed here is better than the ones
discussed in this section. My approach is just parallel in some sense. The
reader interested in foundational issues can find by himself/herself the
right position of this thin note, written by using elementary mathematics.

\ \ \ \ \ \ \ \ \ 

\textbf{Acknowledgements}. This work has been developed within the
activities of the research group in `Theoretical Mechanics' of the `Centro
di Ricerca Matematica Ennio De Giorgi' of the Scuola Normale Superiore at
Pisa. The support of the GNFM-INDAM is acknowledged.

\end{document}